\pdfoutput=1

\documentclass[11pt]{article}

\usepackage{authblk}
\usepackage{url}

\usepackage{subfig}

\usepackage{natbib}
 \bibpunct{(}{)}{;}{a}{,}{,}
\usepackage{graphicx}
\usepackage{amsmath}
\usepackage{amsthm}
\usepackage{diagbox}

\usepackage[pdftex,colorlinks=true,linkcolor=blue,citecolor=blue,urlcolor=blue,bookmarks=false,pdfpagemode=None]{hyperref}
\usepackage[top=1.5in, bottom=1.5in, left=1.5in, right=1.5in]{geometry}

\newtheorem{theorem}{Theorem}
\newtheorem{lemma}{Lemma}
\theoremstyle{remark}
\newtheorem{remark}{Remark}

\DeclareMathOperator*{\argmin}{arg\,min}
\DeclareMathOperator*{\argmax}{arg\,max}
\DeclareMathOperator{\Var}{Var}

\graphicspath{{./plots/}}

\usepackage{setspace}
\begin{document}

\title{Consistent estimation of dynamic and multi-layer\\block models}
\author[1]{Qiuyi Han}
\author[2]{Kevin S.~Xu}
\author[1]{Edoardo M.~Airoldi}
\affil[1]{Department of Statistics, Harvard University, Cambridge, MA, USA
\authorcr \url{qiuyihan@fas.harvard.edu}, \url{airoldi@fas.harvard.edu}}
\affil[2]{Technicolor Research, Los Altos, CA, USA
\authorcr\url{kevinxu@outlook.com}}

\maketitle

\begin{abstract} 
Significant progress has been made recently on theoretical analysis of estimators for the stochastic block model (SBM).
In this paper, we consider the \emph{multi-graph} SBM, which serves as a foundation for many application settings including dynamic and multi-layer networks. 
We explore the asymptotic properties of two estimators for the multi-graph SBM, namely spectral clustering and the maximum-likelihood estimate (MLE), as the number of layers of the multi-graph increases. 
We derive sufficient conditions for \emph{consistency} of both estimators and propose a variational approximation to the MLE that is computationally feasible for large networks. 
We verify the sufficient conditions via simulation and demonstrate that they are practical.
In addition, we apply the model to two real data sets: a dynamic social network and a multi-layer social network with several types of relations.
\end{abstract} 

\section{Introduction}

Modeling relational data arising from networks including social, biological, and information networks has received much attention recently. 
Various probabilistic models for networks have been proposed, including stochastic block models and their mixed-membership variants \citep{airoldi2008mixed,goldenberg2010survey}.
However, in many settings, we not only have a single network, but a collection of networks over a common set of nodes, which is often referred to as a \emph{multi-graph}. 
Multi-graphs arise in several types of settings including dynamic networks with time-evolving edges, such as time-stamped social networks of interactions between people, and multi-layer networks, where edges are measured in multiple ways such as phone calls, text messages, e-mails, face-to-face contacts, etc.

A significant challenge with multi-graphs is to extract common information across the \emph{layers} of the multi-graph in a concise representation, yet be flexible enough to allow differences across layers.
Motivated by the above examples, we consider the \emph{multi-graph stochastic block model} first proposed by \citet{holland1983stochastic}, which divides nodes into classes that define blocks in the multi-graph.
The key assumption is that nodes share the same block structure over the multiple layers, but the class connection probabilities may vary across layers.
We believe this model is a flexible and principled way of analyzing multi-graphs and provides a strong foundation for many applications.
The special case of a single layer, often referred to simply as the stochastic block model (SBM), has been studied extensively in recent years \citep{bickel2009nonparametric,rohe2011spectral,
choi2012stochastic,celisse2012consistency,jin2012fast,
bickel2013asymptotic,amini2013pseudo}.
However, the more general multi-graph case has not been studied as much.

In this paper, we explore the asymptotic properties of several estimators for the multi-graph SBM by letting the number of network layers grow while keeping the number of nodes \emph{fixed}. 
We prove that a spectral clustering estimate of the class memberships is consistent for a special case of the model (Section \ref{sec:spec}).
Next we derive sufficient conditions under which the maximum-likelihood estimate (MLE) of the class memberships is consistent in the general case (Section \ref{sec:mle}).
Finally we propose a variational approximation to the MLE that is computationally tractable and is applicable to many multi-graph settings including dynamic and multi-layer networks (Section \ref{sec:variational}).
We apply the spectral and variational approximation methods to several simulated and real data sets, including both a dynamic social network and a social network with multiple types of relations between people (Section \ref{sec:expri}).

Our main contribution is the consistency analysis for the MLE, which ensures the tractability of the model and paves the way for more sophisticated models and inference techniques. 
To the best of our knowledge, we provide the \emph{first} theoretical results for the multi-graph SBM for a growing number of layers.

\section{Related work}
\label{sec:related_work}
Probabilistic models for networks have been studied for several decades; many commonly used models are discussed in the survey by \citet{goldenberg2010survey}. 
More recent work includes non-parametric network models using graphons \citep{airoldi2013stochastic,wolfe2013nonparametric,gao2014rate}.
Most previous models assume that a single network, rather than a multi-graph, is observed. 

Two settings where multi-graphs arise include dynamic and multi-layer networks.
Dynamic network models typically assume that a sequence of network snapshots is observed at discrete time steps. 
Previous work on dynamic network models has built upon models for a single network augmented with Markovian dynamics. 
\citet{ahmed2009recovering,hanneke2010discrete} built upon exponential random graph models.
\citet{ishiguro2010dynamic,yang2011detecting,ho2011evolving,xu2014dynamic,Xu2015}
built on stochastic block models.
\citet{sarkar2005dynamic,sarkar2007latent,durante2014bayesian} used latent space models. 
\citet{foulds2011dynamic,heaukulani2013dynamic,kim2013nonparametric} used latent feature models.

Multi-layer networks consider multiple types of connections simultaneously. 
For example, Facebook users interact by using ``likes'', comments, messages, and other means. 
Multi-layer networks go by many other names like multi-relational, multi-dimensional, multi-view, and multiplex networks. 
The analysis of multi-layer networks has a long history \citep{holland1983stochastic,
fienberg1985statistical,
szell2010multirelational,
mucha2010community,magnani2011ml,oselio2013multi}.
However there has not been much work on probabilistic modeling of such networks, aside from the multi-view latent space model proposed by \citet{salter2013latent}, which couples the latent spaces of the multiple layers. 

A third related setting involves modeling populations of networks, where each observation consists of a network snapshot drawn from a probability mass function over a network-valued sample space. 
\citet{durante2014nonparametric} proposed a nonparametric Bayesian model for this setting. 
This setting differs from the multi-graph setting that we consider in this paper because the network snapshots (layers) are drawn in an independent and identically distributed (iid) fashion, with no coupling between the snapshots. 

The statistical properties of the inference algorithms in both dynamic and multi-layer network models have not typically been studied. 
Recently there has been a lot of progress on consistency analysis for single networks.
Maximum-likelihood estimation, its variational approximation, and spectral clustering have all been proven to be consistent under the stochastic block model \citep{bickel2009nonparametric,rohe2011spectral,
choi2012stochastic,celisse2012consistency,
zhao2012consistency,jin2012fast,
bickel2013asymptotic,lei2013consistency,Yang:2014fk} as the number of nodes $N \rightarrow \infty$. 
Intuitively, for each new node added to the graph, we observe $N$ realizations, hence larger $N$ provides more information leading to consistent estimation of the model. 

We extend the ideas used in single networks to multi-graphs.
We note that the asymptotic regime is different in this case.
For a single network, one typically lets $N\rightarrow \infty$, while for multi-graphs, we let $T\rightarrow \infty$ with $N$ \emph{fixed}. 
Intuitively it means we do not need to observe a very large network to get a correct understanding of the structure.
Instead, we can gain the information through multiple samples, which may represent, for example, multiple observations over time or multiple relationships.
In practice, it may be more realistic to allow $N$ to grow along with $T$, particularly for the dynamic network setting. 
Allowing $N$ to grow provides \emph{more} information; thus our analysis with fixed $N$ serves as a conservative analysis for different settings.

\section{Multi-graph stochastic block model}
\label{sec:model}
We present an overview of the \emph{multi-graph stochastic block model} first proposed by \citet{holland1983stochastic}.
A single relation is represented by an adjacency matrix ${G^t=(G_{ij}^t), \,i,j=1,\ldots,N}$.
We focus on symmetric binary relations with no self-edges.
For a multi-graph, we observe an \emph{adjacency array} $\vec{G}= \{G^1,\,G^2,\,\ldots ,\, G^T\}$ sharing the same set of nodes.
Subscripts denote the same node pairs for any $t$, while the superscript $t$ indexes layers of the multi-graph.
A layer may refer to time or type of relation depending on the application.
If $\vec{G}$ is a random adjacency array for $N$ nodes and $T$
relations, then the probability distribution of $\vec{G}$ is called a \emph{stochastic multi-graph}.
Let the edge $G_{ij}^t$ be a Bernoulli random variable with success probability $\Phi_{ij}^t$.
${\Phi^t=(\Phi_{ij}^t) \in [0,1]_{N \times N}}$ is the probability matrix of graph $G^t$.
Let ${\vec{\Phi}=\{\Phi^1,\,\Phi^2\,\ldots\, \Phi^T\}}$ be the \emph{probability array}.
We assume the independence of edges within and across layers conditioned on the probability array.
That is, the adjacency array is generated according to
\[
G^t_{ij}|\vec{\Phi} \;\stackrel{\text{ind}}{\sim}\; \text{Bern}(\Phi_{ij}^t).
\]

The multi-graph stochastic block model is a special case of a stochastic multi-graph.
In the multi-graph SBM, networks are generated in the following manner. 
First each node is assigned to a class with probability
$\pi = \{\pi_1, \ldots,\pi_K\}$ where $\pi_k$ is the probability for a node to be assigned to class $k$.
Then, given that nodes $i$ and $j$ are in classes $k$ and $l$, respectively, an edge between $i$ and $j$ in network layer $t$ is generated with probability $P_{kl}^t$.
In other words, nodes in the same classes in the same layer have the same connection probability governed by ${\vec{P} =\{P^1,\,P^2,\ldots ,P^T\}\in [0,1]_{K\times K}}$, the \emph{class connection probability array}.
Let $c_i\in \{1,\ldots,K\}$ denote the class label of node $i$. 
Then $\Phi_{ij}^t=P^t_{c_ic_j}$.

The nodes have class labels $\vec{c}$ \emph{shared} by all of the layers of the multi-graph, but in each layer the class connection probabilities $P^t_{kl}$ may be different.
As we consider undirected networks, $P^t$ is a symmetric matrix with $K(K+1)/2$ free parameters.
One can see that the (single network) SBM is a special case of the multi-graph SBM with $T=1$.

Though simple, this multi-graph model has not been formally studied in the a posteriori setting where class labels are estimated.
It serves as a basis for many settings including dynamic networks and networks with multiple relations.
More importantly, it can be theoretically analyzed and can provide insight on more complex models.

\section{Consistent estimation for the multi-graph stochastic block model}
\label{sec:est}
\citet{holland1983stochastic} only discussed estimation of the multi-graph SBM with blocks specified a priori.
The sample proportion of each layer $t$ is the maximum-likelihood estimate (MLE) of the class probability matrix $P^t$. 
However, in most applications, the block structure is unknown. 
Hence our main goal is to accurately estimate the class memberships.
We extend several inference techniques used for the single network SBM to the multi-layer case.

It is not immediately straightforward how we can utilize inference techniques designed for the single network SBM.
One may imagine inferring $\vec{c}$ independently from each network and averaging across them, e.g.~by majority voting. 
That is, each node is assigned the class label that occurs most often.
We find in simulations that this ad-hoc method often does not work well.

We propose spectral clustering on the mean graph as a motivating method for a special case of the model.
Then we discuss maximum-likelihood estimation, a natural way to combine the information contained in the different layers, for the general case. 
Maximum-likelihood estimation is intractable for large networks so we also consider a variational approximation to the MLE.

Our main focus is on the consistency properties of these methods. 
We consider a fixed number of nodes $N$ but let the number of graph layers $T\rightarrow \infty$. 
In reality, although we do not have infinite layers, we often encounter situations with a large number of layers, such as dynamic networks over long periods of time.

\subsection{Consistency of spectral clustering}
\label{sec:spec}
Spectral clustering is a popular choice for estimating the block structure of the SBM because it scales to large networks and has shown to be consistent as $N \rightarrow \infty$ \citep{Sussman2012}.
The method is based on singular value decomposition and K-means clustering on the singular vectors.

A natural way to extend spectral clustering from single networks to multi-graphs 
is to apply spectral clustering on the mean graph $\bar{G} = \frac{1}{T}\sum_{t=1}^T G^t$.
This method is intuitively appealing as it matches with the assumption of a single set of class labels shared by all of the layers. 
We show that under some stationarity and ergodicity conditions, it indeed provides a consistent estimate of the class assignments.
Specifically we consider the case where the class connection probabilities $P^t$ vary across layers but have the same mean $M$.
The following theorem shows the consistency of spectral clustering on the mean graph $\bar{G}$ if the mean $M$ is identifiable.

\begin{theorem}
\label{theorem:spectral}
Assume $\vec{P}$ follows a stationary ergodic process such that $E(P_{kl}^t) = \mu_{kl}$ and $Var(P_{kl}^t) = \epsilon_{kl}^2$ for all $t$. 
Assume $M = [\mu_{kl}]$ is identifiable, i.e.~$M$ has no identical rows.
Let $\bar{G} = \frac{1}{T}\sum_{t=1}^T G^t$. 
Spectral clustering of $\bar{G}$ gives accurate labels as $T\rightarrow \infty$. That is, let $U_{N\times K}$ be the first K right singular vectors in the singular value decomposition of $\bar{G}$. K-means clustering on the rows of $U_{N\times K}$ outputs class estimates $\hat{c}_1,...,\hat{c}_N$. Up to permutation, 
$\hat{c} = c, \text{ a.s. as } T\rightarrow \infty$.
\end{theorem}
We provide a sketch of the proof; details can be found in Appendix \ref{sec:ProofSpectral}. 
Since we have independent errors in the probability matrix and also independent errors in the Bernoulli observations, averaging cancels the error so that 
$\bar{G}\rightarrow CMC'$. 
Here $C$ is a rank $K$ matrix incorporating the class assignment vectors. 
Using an inequality from \citet{oliveira2009concentration}, we bound the distance between the singular vectors of $\bar{G}$ and $CMC'$. 
Therefore, spectral clustering on $\bar{G}$ clusters the nodes into $K$ different classes.

\begin{remark}
To determine the number of classes is a difficult model selection problem even for a single network. 
We will not discuss this problem in detail. 
We assume $K$ is fixed and known in this paper.
\end{remark}

\begin{remark}
The diagonal of $G^t$ is always $0$ because no self-edges are allowed; however the diagonal of $CMC'$ is not necessarily $0$. 
This may not cause a problem as $N\rightarrow \infty$. 
But for finite $N$, it may cause error in estimating the singular vectors. 
If this is the case, we may utilize the singular value decomposition that minimizes the off-diagonal mean square error
\[
\argmin_{U,S} \sum_{i < j} (\bar{G}_{ij}-U_iSU_j')^2,
\]
which can be computed by iterative singular value decomposition \citep{scheinerman2010modeling}. 
\end{remark}

The condition in Theorem \ref{theorem:spectral} requiring $\vec{P}$ to be stationary with identifiable mean $M$ is restrictive. 
Spectral clustering on the mean graph is not effective in many cases. 
Consider the case for which
\[
P^t \in \left\{
\begin{pmatrix}
0.7 & 0.3 \\
0.3 & 0.7
\end{pmatrix},
\begin{pmatrix}
0.3 & 0.7 \\
0.7 & 0.3
\end{pmatrix}
\right\}
\]
where both outcomes are equally likely for all $t$.
Then
\[
M = \begin{pmatrix}
0.5 & 0.5 \\
0.5 & 0.5 \end{pmatrix}
\]
is not identifiable.
Spectral clustering on the mean graph fails to correctly estimate the class assignments as $T \rightarrow \infty$.
But there is information contained in this multi-graph. 
We can use maximum likelihood estimation to estimate the class assignments correctly in this case.

\subsection{Consistency of maximum likelihood estimate}
\label{sec:mle}
Now we focus on the general case where we do not place any structure on $\vec{P}$.
A natural way to estimate the class assignment is to use the maximum-likelihood estimate (MLE). 
We show that for a large enough fixed $N$, the MLE will estimate the class memberships correctly as $T \rightarrow \infty$.

First we define some notation. 
For any class assignment $z$, let $n_k(z)=\#\{i:z_i=k\}$ be the number of nodes in class $k$. 
Let $m(z)=\min_k n_k(z)$ denote the minimum number of nodes in any class under labels $z$.
Let 
\[
n_{kl}(z)=\begin{cases} n_kn_l, & \, k \neq l \\
			n_k(n_k-1)/2, & \, k= l. \end{cases}
\]
be the number of pairs of nodes in each block.
We drop the dependency on $z$ whenever it is unambiguous. We also drop the superscript $t$ when we talk about a single layer of the network.

Now define some notation related to the MLE. The complete log-likelihood for parameters $(z,P)$ is 
\[
l(z,P)=\sum_{i<j} \left(G_{ij}\log(P_{z_iz_j})+(1-G_{ij})\log(1-P_{z_iz_j})\right).
\]
Here $P$ is a parameter not to be confused with the true class connection probability matrix. In particular, we are interested in the case $P_{kl}=\bar{P}_{kl}(z)$ where 
\[
\bar{P}_{kl}(z)=\frac{1}{n_{kl}(z)}\sum_{\substack{i:z_i=k \\ j:j\neq i,z_j=l}} P_{c_ic_j}.
\]
Here $\bar{P}$ is the average of the true $P$ under block assignment $z$. To ease notation, let
\begin{equation*}
\sigma(p)=p\log(p)+(1-p)\log(1-p). \label{eqn:sigma}
\end{equation*}
Denote the expectation of log-likelihood of $(z,\bar{P}(z))$ as 
\begin{equation*}
h(z)=E[l(z,\bar{P}(z))]=\sum_{k\leq l} n_{kl}(z) \, \sigma(\bar{P}(z)). \label{eqn:h}
\end{equation*}
Now, as we do not observe the true $P$, the natural step is to estimate it with the empirical mean for any given $z$. So let 
\[
o_{kl}(z)=\sum_{\substack{z_i=k\\z_j=l}} \begin{cases} G_{ij},&\,k\neq l\\
	G_{ij}/2, & \, k= l. \end{cases}
\]
be the observed number of edges in block $(k,l)$. 
Then the profile log-likelihood \citep{bickel2009nonparametric} is defined as
\begin{equation*}
f(z)=\sum_{k\leq l} n_{kl}(z) \, \sigma \!\left(\frac{o_{kl}(z)}{n_{kl}(z)}\right). \label{eqn:f}
\end{equation*}
Let the expectation of $f$ be
\begin{equation*}
g(z)=E(f(z))=\sum_{k\leq l} n_{kl}(z) \,E\!\left[\sigma\left(\frac{o_{kl}(z)}{n_{kl}(z)}\right)\right]. \label{eqn:g}
\end{equation*}
Now we are ready to state the consistency of the MLE for the multi-graph SBM. 
If all elements of $P^t$ are bounded away from $0$ and $1$ and their column differences are at least some distance apart, then when we have a sufficient number of nodes in each block, the true label $c$ uniquely maximizes the sum of profile log-likelihoods over the layers.
\begin{theorem}
\label{theorem:consistency}
Let 
\begin{gather*}
C_0=\inf_{t,k,l} (P^t_{kl},1-P^t_{kl}) \\
\delta=\inf_{t,k,l} \max_m \bigg[\sigma(P^t_{km})+\sigma(P^t_{lm}) 
-2\sigma\left(\frac{P^t_{km}+P^t_{lm}}{2}\right)\bigg].
\end{gather*}
Assuming $C_0>0$ and $\delta>0$,
if $m(c) = \min_k n_k(c)$ is sufficiently large, then
\[
\hat{c} = \argmax_z \sum_t f^t(z)\rightarrow c, \text{ a.s. as } T\rightarrow \infty.
\]
\end{theorem}
The idea is that $\sum_t f^t(z)$ is the sum of independent profile log-likelihoods. 
We need $N$ to be sufficiently large so that the expectation of the profile log-likelihood  at each layer is maximized at the true labels $c$. 
Then as $T\rightarrow \infty$, we have convergence to expectation for $\sum_t f^t(z)$.
We formalize the ideas by establishing the following lemmas.

\begin{lemma}[From \citet{choi2012stochastic}]
\label{lemma:KLdiffernece}
For any label assignment $z$, let $r(z)$ count the number of nodes whose true class assignments under $c$ are not in the majority within their respective class assignment under $z$. Let 
\[
\delta=\min_{k,l} \max_m \sigma(P_{km})+\sigma(P_{lm})-2\sigma\!\left(\frac{P_{km}+P_{lm}}{2}\right).
\]
Then the expectation of the log-likelihood $h$ is maximized by $h(c)$, and
\[
h(c)-h(z)\geq \frac{r(z)}{2}\delta\min_k n_k(c). 
\]
In particular, for all $z\neq c$,
\[
h(c)-h(z)\geq \frac{1}{2}\delta\min_k n_k(c). 
\]
\end{lemma}
Lemma \ref{lemma:KLdiffernece} shows the expectation of the log-likelihood is maximized at the true parameters, and the difference of the true parameters and any other candidate is at least some distance apart which depends on the column difference of the probability matrix.
However, as we work with the profile log-likelihood, we establish Lemmas \ref{lemma:sigma} and \ref{lemma:profile-mle} to bound the difference between the expectation of the profile log-likelihood and the complete log-likelihood.

\begin{lemma}
\label{lemma:sigma}
Let $x\sim \frac{1}{N}\text{Bin}(N,p)$.
For $p\in (0,1)$, 
\[
E(\sigma(x))\rightarrow \sigma(p)+\frac{1}{2N} + O\left(\frac{1}{N^2}\right), \text{ as } N\rightarrow \infty.
\]
\end{lemma}
\begin{lemma}
\label{lemma:profile-mle}
Assume $C_0\leq P_{kl}\leq 1-C_0$, $C_0>0$. For any $\delta_0>0$ and any $z$, if $\min_k n_k(z)$ is large enough, then
the difference between the expectation of the profile log-likelihood $g(z)$ and the expectation of the complete log-likelihood $h(z)$ is bounded in the following manner:
\[
\left|g(z)-h(z)-\frac{K(K+1)}{4}\right| \leq  \delta_0.
\]
\end{lemma}
Lemma \ref{lemma:sigma} utilizes a Taylor series expansion. 
For simplicity, we use big $O(\cdot)$ notation instead of specifying an actual bound. 
Readers can refer to Appendix \ref{sec:ProofML} for the bound and the constants in the bound.
Lemma \ref{lemma:profile-mle} uses Lemma \ref{lemma:sigma} and shows that with a sufficiently large number of nodes, the difference of the expectation of the profile log-likelihood and complete log-likelihood is $K(K+1)/4$ and a negligible term $\delta_0$.
Combining the lemmas and using the concentration inequality, we can show that Theorem \ref{theorem:consistency} provides sufficient conditions for the consistency of the multi-graph SBM. 
The proofs of Lemmas \ref{lemma:KLdiffernece}--\ref{lemma:profile-mle} and Theorem \ref{theorem:consistency} can be found also in Appendix \ref{sec:ProofML}. 

\begin{remark}
The main difference between the $N\rightarrow \infty$ case considered in most previous work and the $T\rightarrow \infty$ case that we consider is that, for $N\rightarrow \infty$, a direct bound is put on $f$ and $l$. For $T\rightarrow \infty$, we need only to bound the \emph{expectation} of $f$ and $l$. This is newly studied here.
In other words, for some particular class connection probability matrix $P$, the number of nodes required in a single network to have an accurate estimate is much larger than what is needed in a multi-graph with a growing number of layers.
\end{remark}

\subsubsection{Variational approximation}
\label{sec:variational}
The MLE is computationally infeasible for large networks because the number of candidate class assignments grows exponentially with the number of nodes.
To overcome the computational burden, variational approximation, which replaces the joint distribution with independent marginal distributions, can be used to approximate the MLE.
\citet{daudin2008mixture} has a detailed discussion of variational approximation in the SBM.
We adapt it to the multi-graph SBM, resulting in the following update equations:
\begin{gather*}
b_{ik} \propto \pi_k \prod_{j \neq i} \prod_{t} \prod_l \left[(P_{kl}^t)^{g_{ij}^t} (1-P_{kl}^t)^{1-g_{ij}^t}\right]^{b_{jl}} \\
\pi_k \propto \sum_i b_{ik} \\
P_{kl}^t = \frac{\sum_{i \neq j} b_{ik} b_{jl} g_{ij}^t} {\sum_{i \neq j} b_{ik} b_{jl}},
\end{gather*}
where the $b_{ik}$'s denote the variational parameters.
The derivation is straightforward; we provide details in Appendix \ref{sec:Variational}.

Variational approximation has been shown to be consistent in the SBM \citep{celisse2012consistency,bickel2013asymptotic}.
We conjecture that the performance of variational approximation is also good in the multi-graph SBM.
Unless otherwise specified we use variational approximation to replace the MLE in all experiments. 

\begin{figure*}[tp]
\centering
\subfloat[$p=0.1$]{\includegraphics[width=2.5in]{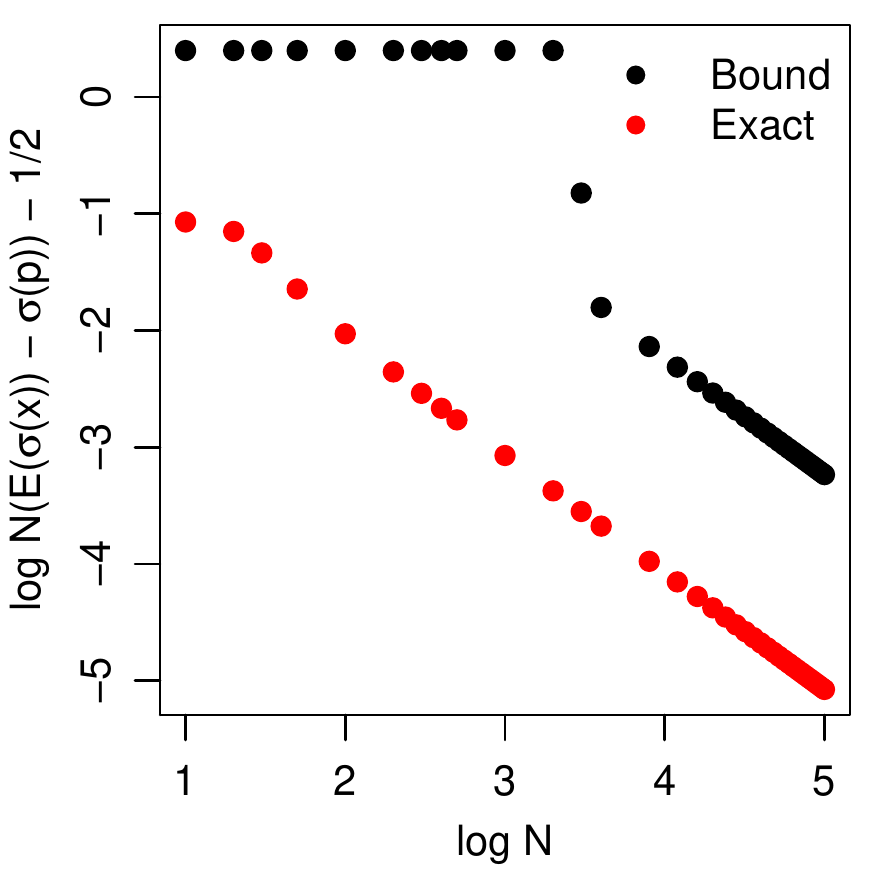}} \qquad
\subfloat[$p=0.25$]{\includegraphics[width=2.5in]{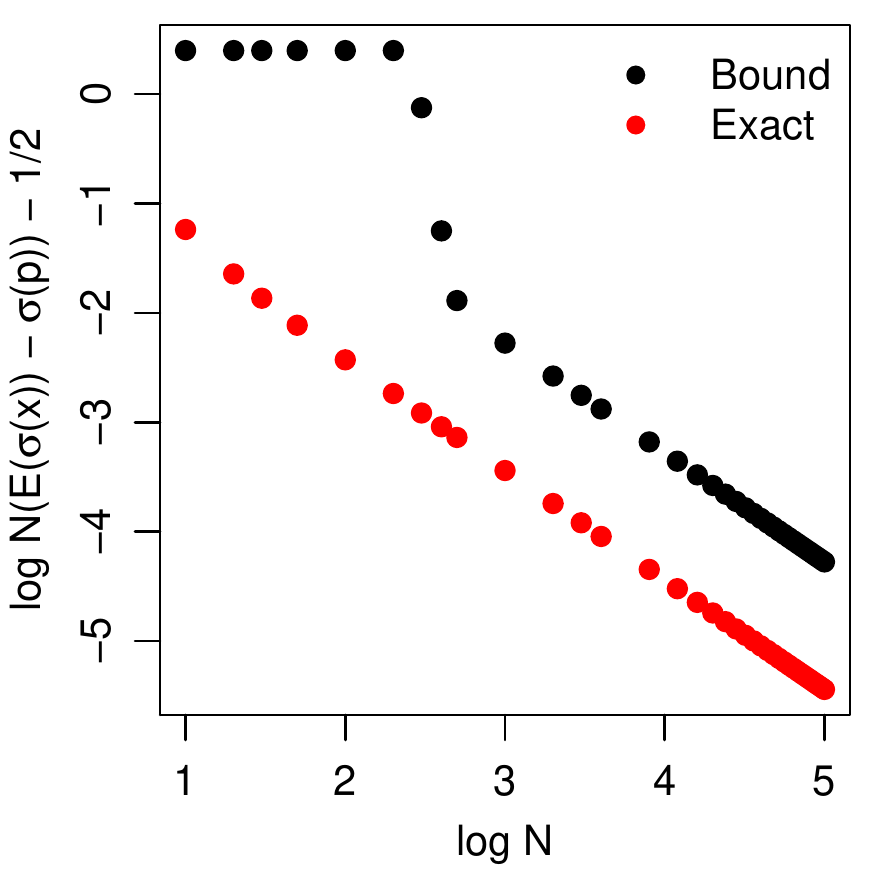}} \qquad
\subfloat[$p=0.4$]{\includegraphics[width=2.5in]{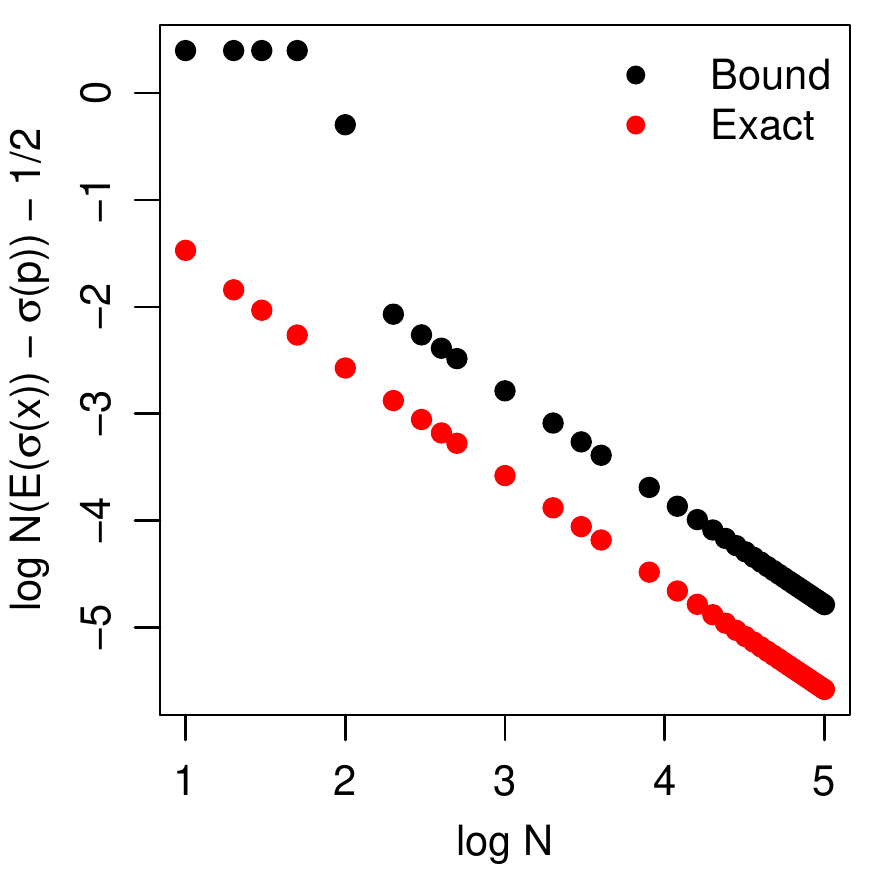}}
\caption{Comparison of bound in Lemma \ref{lemma:sigma} to exact values of $N(E(\sigma(x))-\sigma(p))-1/2$ for varying $N$ and $p$. The tightness of the bound affects the minimum number of nodes required to guarantee consistency in Theorem \ref{theorem:consistency}.}
\label{figure:bounds}
\end{figure*}

\section{Experiments}
\label{sec:expri}

\subsection{Numerical illustration}

We begin with a toy example where we investigate empirically how many nodes are needed for the profile MLE to correctly recover the classes as $T \rightarrow \infty$. 
Due to the computational intractability of computing the exact profile MLE, we consider very a small network with $N=16$ nodes and $K=2$ classes where each class has $8$ nodes. Consider two multi-graph SBMs with the following probability matrices:
\[
\text{Case 1: }P^t \equiv \left( \begin{array}{cc}
0.55 & 0.45 \\
0.45 & 0.55 \end{array} \right)
\]
\[
\text{Case 2: }
P^t \equiv \left( \begin{array}{cc}
0.51 & 0.49 \\
0.49 & 0.51 \end{array} \right)
\]

The $\delta$ (defined in Theorem \ref{theorem:consistency}) corresponding to the row difference of $P^t$ is much smaller in case 2. 
Empirically the profile MLE succeeds to get the true labels in case 1 while it fails in case 2.
Further analysis shows that in order to have consistency given the class connection probability matrix $P^t$ in case 2, the total number of nodes should be at least $40$.
This toy example demonstrates that conditions on the probability matrices and network size are necessary for consistency. 
Theorem \ref{theorem:consistency} provides \emph{sufficient} conditions. 

Next we investigate the tightness of the conditions in Theorem \ref{theorem:consistency}. 
The tightness of Lemma \ref{lemma:KLdiffernece} was studied by \citet{choi2012stochastic}. 
We check the tightness of Lemma \ref{lemma:sigma}. 
For different $p$, we can calculate the exact value of $N(E(\sigma(x))-\sigma(p))-1/2$ and compare it to the bound from Lemma \ref{lemma:sigma}. 
Figure~\ref{figure:bounds} shows that the bound is loose for small $N$, but has almost the same asymptotic decay as the exact calculation.
For small $N$, the remainder in Taylor expansion causes deviation. 
Also the bounds are looser for $p$ closer to $0$ or $1$ but still informative in most cases. 

\begin{table}[t]
\caption{Minimum number of nodes $N$ required for consistency of the profile MLE with $K=2$ classes under different values for parameters $C_0$ and $\delta$ from Theorem \ref{theorem:consistency}.}
\label{tab:MinNodes}
\centering
\begin{tabular}{|c|cccccc|}
  \hline
  \diagbox{$\delta$}{$C_0$}& 0.3 & 0.25 & 0.2 & 0.15 & 0.1 & 0.05 \\ 
  \hline
  0.165 &  42 &  50 &  64 &  88 & 124 & 184 \\ 
  0.091 &  44 &  52 &  66 &  92 & 142 & 234 \\ 
  0.040 &  46 &  56 &  70 &  94 & 148 & 314 \\ 
  0.010 &  66 &  68 &  74 & 100 & 156 & 330 \\ 
  \hline
\end{tabular}
\end{table} 

For the special case of $K=2$ classes, we can calculate all of the constants in the sufficient conditions in Theorem \ref{theorem:consistency} for different values of $C_0$ and $\delta$ by enumerating cases. 
Details are provided in Appendix \ref{sec:MinNodes}. 
Table \ref{tab:MinNodes} shows the smallest number of nodes $N$ that is sufficient for consistency of the profile MLE to hold for different values of $C_0$ and $\delta$. 
Note that the minimum $N$ is in the tens or hundreds, suggesting that the bounds in Theorem \ref{theorem:consistency} are not overly loose and are indeed of practical significance. 

\subsection{Comparison with majority voting}

As previously mentioned, majority voting is another way to utilize inference methods for a single network on multi-graphs. 
We consider two majority vote methods as baselines for comparison, one that utilizes spectral clustering on each layer, and one that applies a variational approximation to each layer.
When using majority voting between different layers of the network, the estimated class labels for each layer must first be aligned or matched. 
We utilize the Hungarian algorithm \citep{Kuhn1955} to compute the maximum agreement matching between the estimated labels at layer $t$ with the majority vote up to layer $t-1$. 

We conduct simulations to compare our proposed methods of spectral clustering on the mean graph and profile maximum-likelihood estimation with the majority vote baselines. 
We consider a well-studied scenario where we have $128$ nodes initialized randomly into $4$ classes \citep{newman2004finding}. 
For each layer, the within-class connection probability is $0.0968$, and the between-class connection probability is $0.0521$. 
Under such connection probabilities, the classes are below the detectability limit  \citep{Decelle2011} for a single layer, so the class estimation accuracy from a single layer is very low. 
We increase the number of layers and observe how the accuracy changes. 

\begin{figure}[t]
\centering
\includegraphics[width=4in]{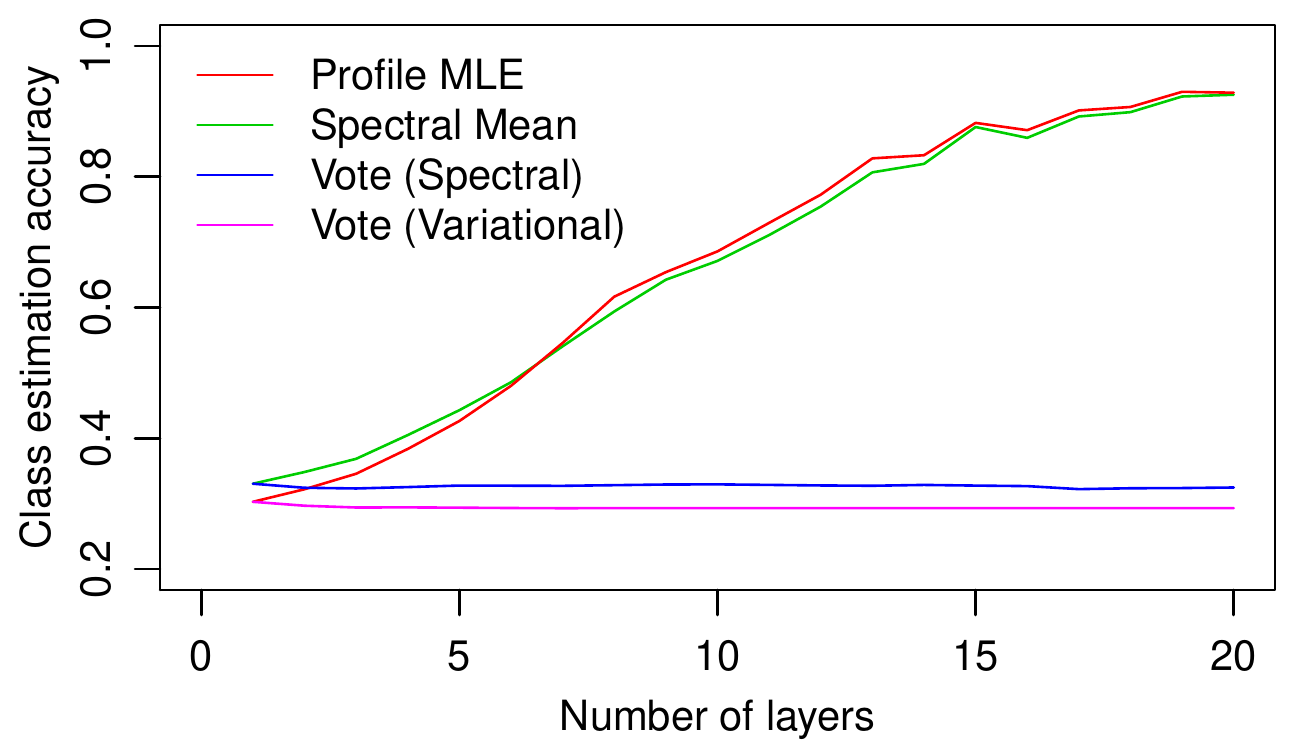}
\caption{Simulation experiment comparing the proposed methods of profile MLE and spectral clustering on the mean graph with two majority vote baselines. The proposed methods increase in accuracy as the number of layers increases, but the two heuristic methods based on majority vote do not.}
\label{figure:majorityVote}
\end{figure}  

Figure~\ref{figure:majorityVote} shows the accuracy of the two proposed methods compared to the two majority voting methods averaged over $100$ replications. 
Both the profile MLE and spectral clustering on the mean graph have the anticipated increasing accuracy over time. 
But the accuracies of the two heuristic majority vote methods do not improve. 
Though one may expect the errors in majority vote to be canceled out over time, these results show that, without careful averaging of errors, we cannot gain from the multiple layers. 
We find that this is due to choosing connection probabilities below the detectability limit; if we make the estimation problem easier by increasing the within-class probability above the detectability limit, then the majority vote methods do improve with increasing layers, albeit much slower than the methods we propose in this paper.

\subsection{MIT Reality Mining data} 
Next we apply our model on the MIT Reality Mining data set \citep{eagle2006reality}. 
This data set comprises $93$ students and staff at MIT in the 2004-2005 school year during which time their cell phone activities were recorded. 
We construct dynamic networks based on physical proximity,
which was measured using scans for nearby Bluetooth devices
at $5$-minute intervals. 
We exclude data near the beginning and end of the experiment where participation was low. 
We discretize time into $1$-week intervals, similar to \citet{Mutlu2012,Xu2014Adaptive}, resulting
in $39$ time steps between August 2004 and May 2005.

\begin{figure}[t]
\centering
\includegraphics[width=4in]{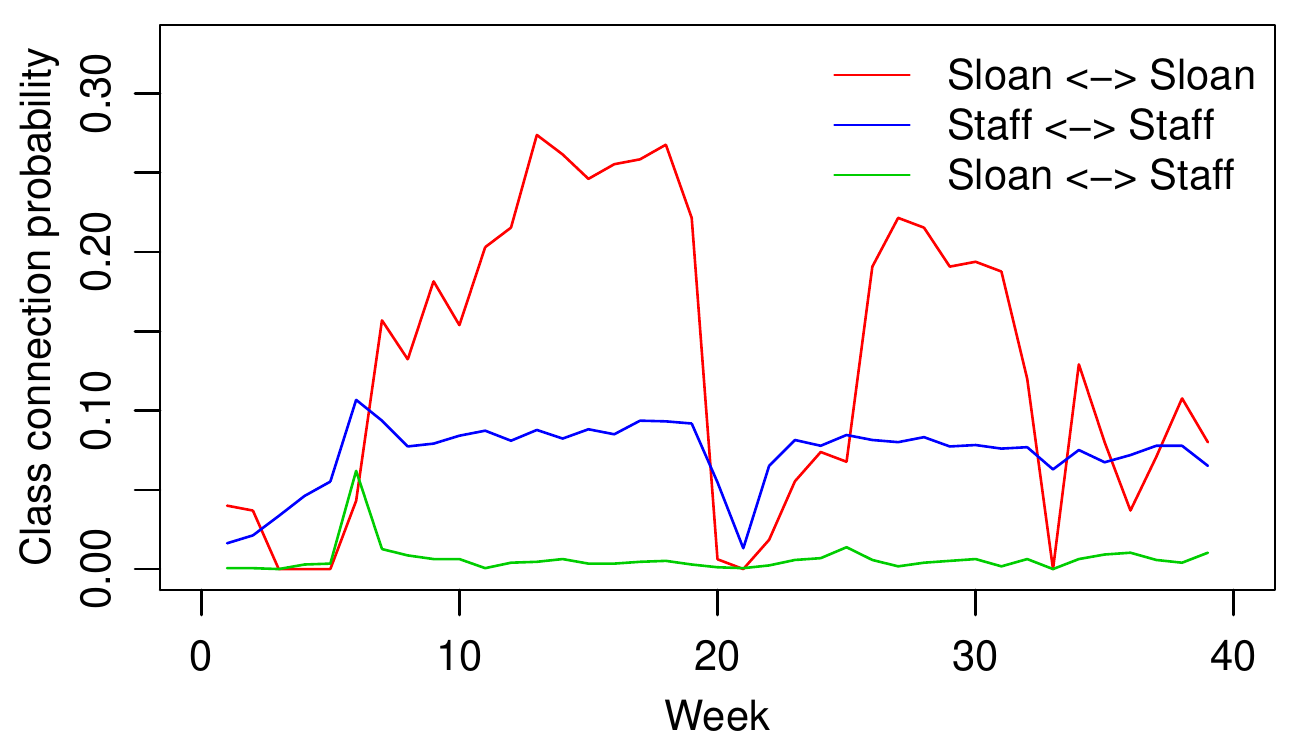}
\caption{Estimates of class connection probabilities in the Reality Mining data set. The probabilities vary significantly over time, particularly for edges between Sloan students.}
\label{figure:RealityMining}
\end{figure}

\begin{table}[t]
\caption{Class estimation accuracy in the Reality Mining data set given data up to week listed in the first column. 
Best performer in each row is listed in bold. 
Both the proposed spectral clustering on the mean graph and profile maximum-likelihood estimation approaches improve over time, but majority vote does not.}
\label{tab:RealityMining}
\centering
\begin{tabular}{|cccc|}
  \hline
  Week & Maj.~vote & Spectral Mean & Profile MLE \\ 
  \hline
  10 & {\bf 0.76} & 0.62 & 0.57 \\ 
  15 & 0.82 & 0.94 & {\bf 0.95} \\ 
  20 & 0.83 & 0.95 & {\bf 0.98} \\ 
  25 & 0.78 & 0.95 & {\bf 0.99} \\ 
  30 & 0.80 & 0.97 & {\bf 0.99} \\ 
  35 & 0.80 & 0.97 & {\bf 0.99} \\ 
  End & 0.77 & 0.97 & {\bf 0.99} \\ 
  \hline
\end{tabular}
\end{table}

We treat the affiliation of participants as ground-truth class labels and test our proposed methods. 
Two communities are found: one of $26$ Sloan business school students, and one of $67$ staff working in the same building.
Since degree heterogeneity may cause problems in detecting communities using the SBM \citep{karrer2011stochastic}, we reduce its impact by connecting each participant to the $5$ other participants who spent the most time in physical proximity during each week. 
Figure~\ref{figure:RealityMining} shows the empirical block connection probabilities within and between the two classes, estimated by the profile MLE.
The class connection probabilities vary significantly over time, which validates the importance of the varying class connection probability assumption in our model. 
Notice that the two communities become well-separated around week $8$. 
The class estimation accuracies for the different methods are shown in Table~\ref{tab:RealityMining}. 
Since the community structure only becomes clear at around week $8$, the spectral and profile MLE methods are initially worse than majority voting but quickly improve and are superior over the remainder of the data trace. 
By combining information across time, the proposed methods successfully reveal the community structure while majority voting continues to improperly estimate the classes 
of about $20\%$ of the people. 

\subsection{AU-CS multi-layer network data}

\begin{figure*}[tp]
\newsavebox{\tempbox}
\centering
\subfloat[Co-authorship]{\includegraphics[scale=0.28]{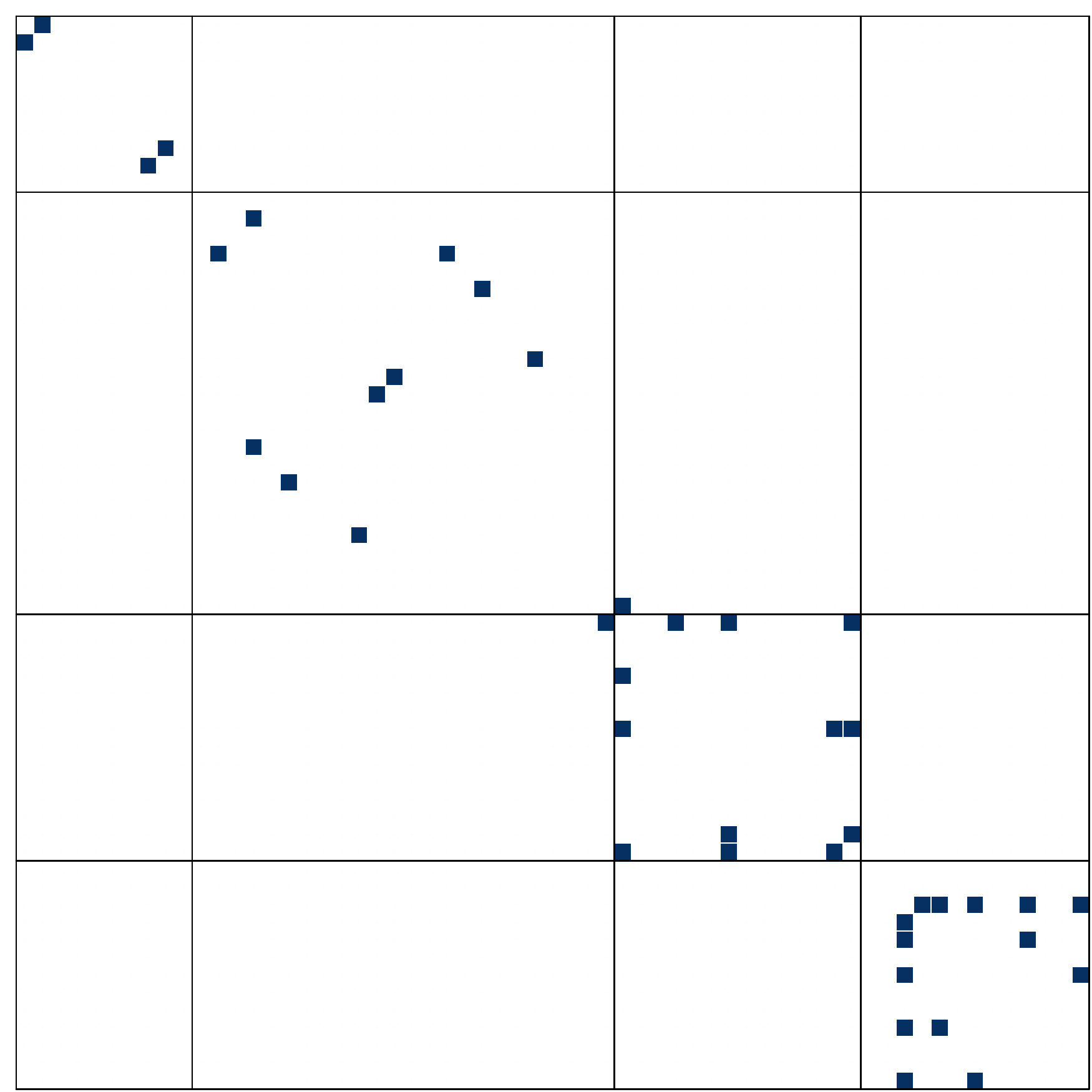}} \qquad\qquad
\subfloat[Facebook]{\includegraphics[scale=0.28]{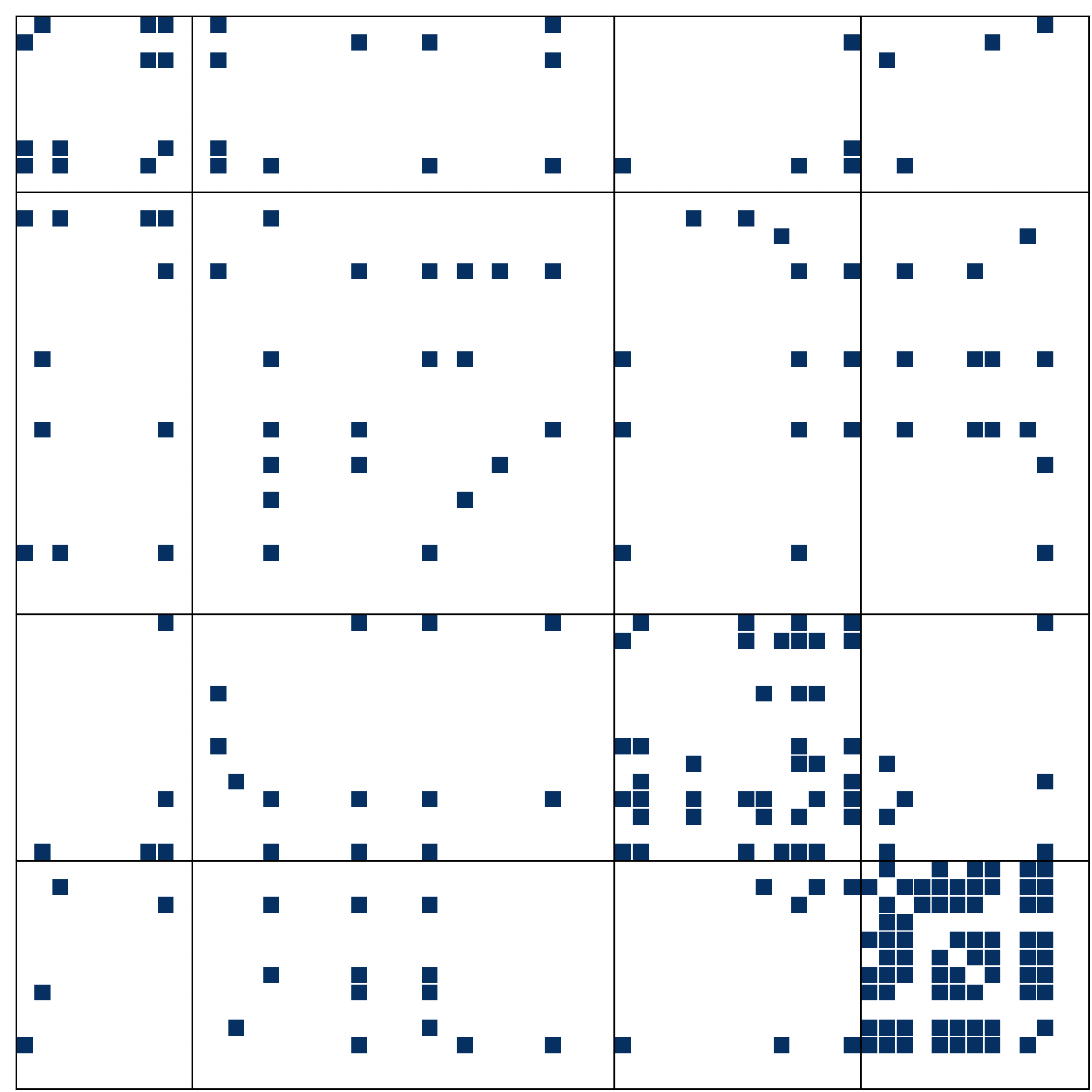}} \qquad\qquad
\subfloat[Leisure]{\includegraphics[scale=0.28]{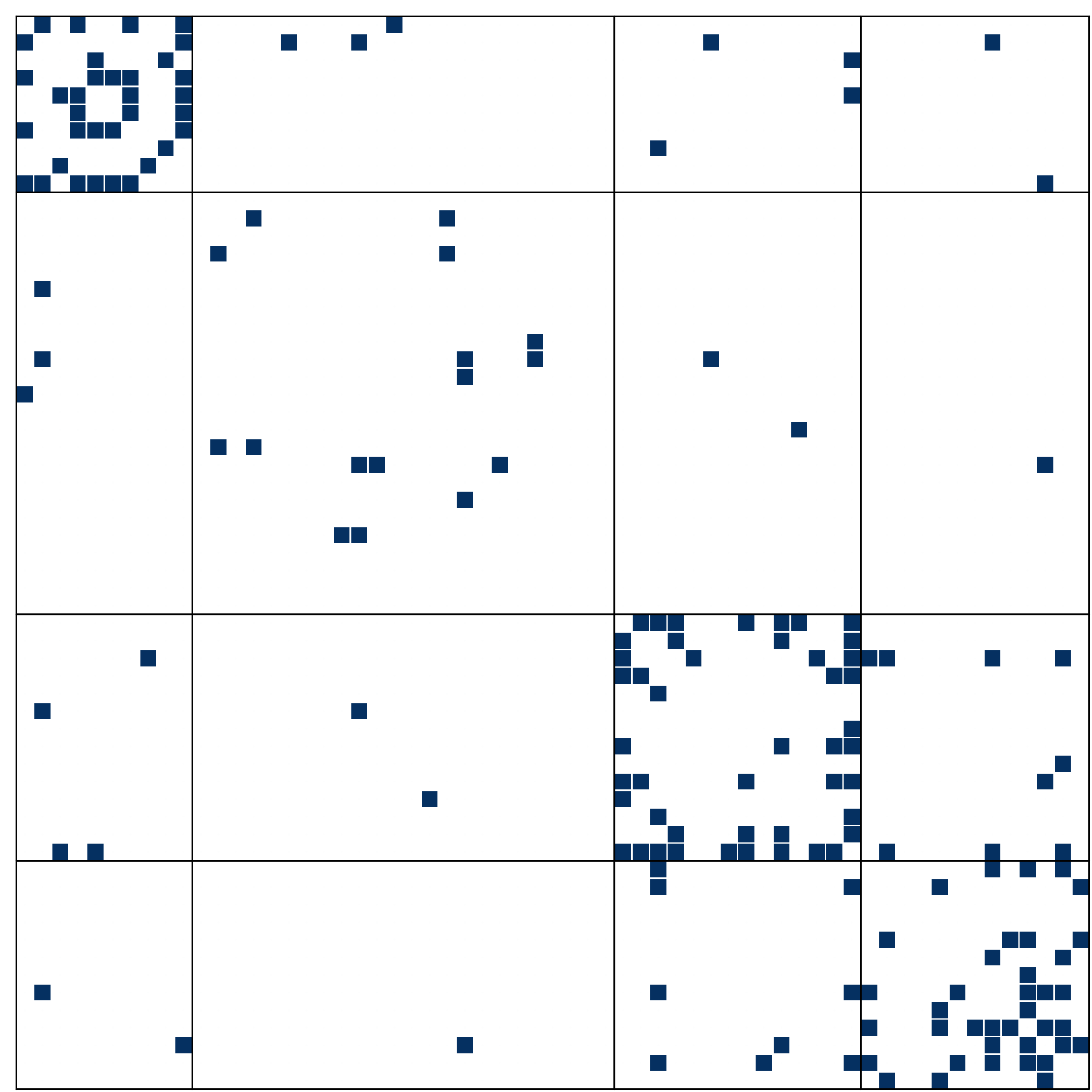}} \qquad\qquad
\subfloat[Lunch]{\includegraphics[scale=0.28]{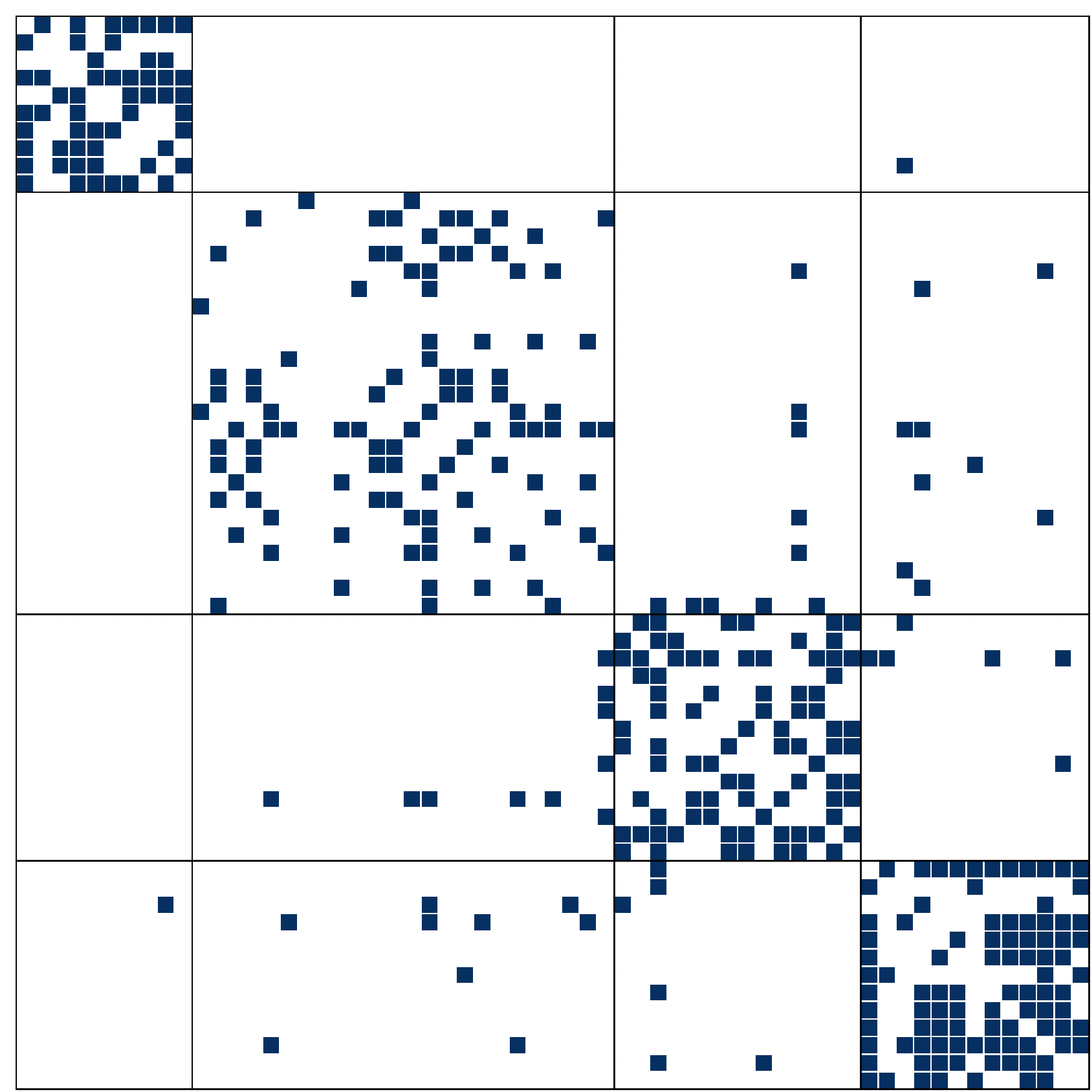}} \qquad\qquad
\sbox{\tempbox}{\includegraphics[scale=0.28]{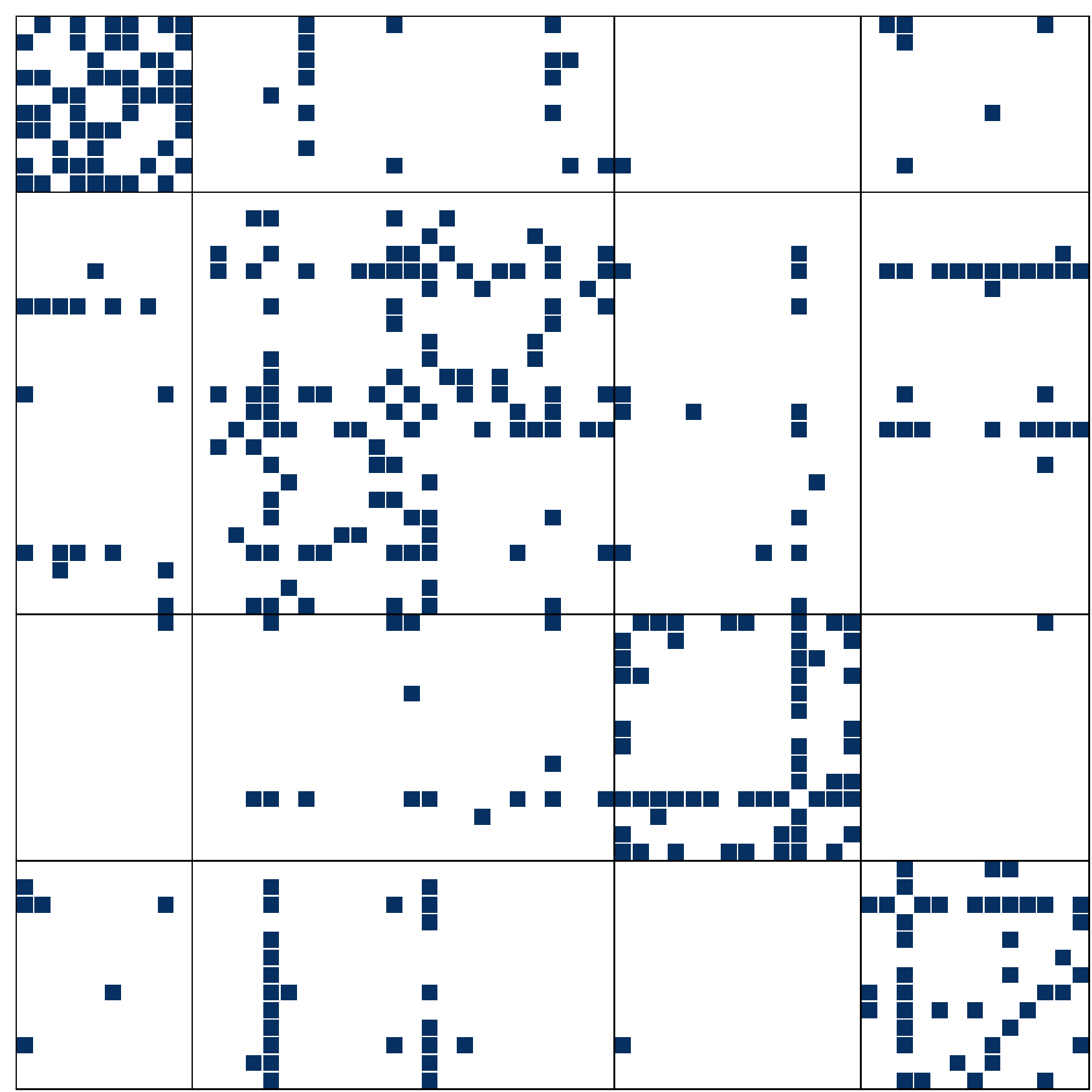}}%
\subfloat[Work]{\usebox{\tempbox}}%
\qquad\qquad
\subfloat[ICL (lower denotes better fit)]{
\parbox[b]{1.88in}{
\centering
\begin{small}
\begin{sc}
\renewcommand\arraystretch{1.2}
\renewcommand\tabcolsep{12pt}
\begin{tabular}[b]{|cc|}
\hline
$K$ & ICL \\
\hline
2 & 4087 \\
3 & 3914 \\
{\bf 4} & {\bf 3830} \\
5 & 3841 \\
6 & 3878 \\
\hline
\end{tabular}
\end{sc}
\end{small}
\\
\vspace{24pt}
}
}
\caption{The estimated community structures in the AU-CS multi-layer networks overlaid onto the adjacency matrices of different relations. The dots denote connections (edges), and the grids correspond to SBM blocks.}
\label{figure:multilayer}
\end{figure*}

We look at another example from a multi-layer network comprising five kinds of self-reported on-line and off-line relationships between the employees of a research department: Facebook, leisure, work, co-authorship, and lunch \citep{aucs}.
We assume the class structure to be invariant across the different types of relations and apply our model.
For model selection, we extend the Integrated Classification Likelihood (ICL) criterion proposed by \citet{daudin2008mixture} for the single network SBM to multi-graphs to select the number of blocks $K$. 
Specifically we maximize
\begin{equation*}
-2Q(\vec{G})+(K-1)\log N + \frac{TK(K+1)}{2} \log \frac{N(N-1)}{2},
\end{equation*}
where $Q(\vec{G})$ is the variational approximation to the complete log-likelihood.
We initialize the variational approximation with different randomizations as well as the spectral clustering solution. 
The term is maximized at $K=4$.

Figure \ref{figure:multilayer} shows the estimated $4$ classes overlaid onto the adjacency matrix of each relation.
Although we have no ground truth for this data set, we detect well-separated communities in all relations aside from co-authorship, which is an extremely sparse layer. 
Notice once again the difference in empirical connection probabilities over the multiple layers of the multi-graph. 

For this data set, we do not have ground truth labels to evaluate the class estimation accuracy. 
We note, however, that the ICL obtained by our variational approximation algorithm is much better than the ICLs obtained by fitting an SBM on the mean graph and by majority vote, both of which are over $4000$.

\section{Discussion}
\label{sec:dis}
In this paper, we investigated the multi-graph stochastic block model applied to dynamic and multi-layer networks with invariant class structure.
Both spectral clustering on the mean graph and maximum-likelihood estimation are proved to be consistent for a fixed number of nodes when we have an increasing number of network layers, provided certain sufficient conditions are satisfied. 

There are several interesting avenues for extensions of our analysis. 
First we can add a layer of probabilistic modeling on the probability matrices if we have additional information. 
Since dynamic networks usually vary smoothly over time, we can put a state-space model on the adjacency array \citep{xu2014dynamic}. 
We can also use a hierarchical model on the probability matrices to couple them for analyzing multi-layer networks. 
Since our sufficient conditions do not consider such additional structure, an interesting area of future work would be to derive sufficient conditions that utilize the structure on the probability matrices, which would likely produce tighter bounds. 
It would also be interesting to draw connections to recent work on consistent estimation for populations of networks \citep{durante2014nonparametric}, for which no coupling between samples (layers) exists. 

\appendix

\section{Proof of Theorem \ref{theorem:spectral}}
\label{sec:ProofSpectral}
\begin{proof}
Denote the latent class label for each node as a vector $\vec{C}_i=(C_{i1},...,C_{iK})$ where
\begin{equation*}
C_{ij}=\begin{cases} 0, & \, c_i\neq j \\
			1, & \, c_i=j \end{cases}.
\end{equation*} 
Define the $N\times K$ matrix
\begin{equation*}
C= \begin{pmatrix}
\vec{C}_1 \\
\vdots \\
\vec{C}_N 
\end{pmatrix}. 
\end{equation*}
Notice that 
\[
E(\bar{G}) = E (\bar{\Phi}) = CE(\bar{P})C'=CMC',
\]
where the last equality follows from ergodicity of the process $\{P^t\}$. 
Intuitively $\bar{G}$ would converge to $CMC'$.
Since the matrix of eigenvectors of $CMC'$ only has $K$ distinct rows, the eigenvectors of $\bar{G}$ would converge to those of $CMC'$, and eventually the rows of the eigenvector matrix would be well-separated for nodes in different classes. 

More formally, we first bound the difference of $\bar{G}$ and $CMC'$.
Let $\|A\|_F = (\sum_{i,j} a_{ij}^2)^{1/2}$ denote the Frobenius norm of a matrix $A$. 
We have 
\begin{equation*}
E(\|\bar{G}-CMC'\|_F^2) = \sum_{i,j} \Var (\bar{G}_{ij})=\sum_{i,j} \left[E(\Var(\bar{G}_{ij}|\bar{\Phi}_{ij}))+\Var(E(\bar{G}_{ij}|\bar{\Phi}_{ij}))\right].
\end{equation*}
The first term can be bounded by
\begin{equation*}
\Var(\bar{G}_{ij}|\bar{\Phi}_{ij}) =  \frac{\bar{\Phi}_{ij}(1-\bar{\Phi}_{ij})}{T} \leq \frac{1}{4T}
\end{equation*}
because $0 \leq \bar{\Phi}_{ij} \leq 1$. 
For the second term,
\begin{equation*}
\Var(E(\bar{G}_{ij}|\bar{\Phi}_{ij})) =\Var(\bar{\Phi}_{ij})
=\Var(\bar{P}_{c_i c_j})
=\frac{\epsilon_{c_i c_j}^2}{T}.
\end{equation*}
Therefore, 
\[
E(\|\bar{G}-CMC'\|_F^2) \leq \frac{N^2(1+4\epsilon^2)}{4T}
\]
where $\epsilon = \displaystyle\max_{c_i,c_j} \epsilon_{c_i c_j}$. 
By the Markov inequality, for any $\delta > 0$,
\[
P(\|\bar{G}-CMC'\|_F^2>\delta)\leq \frac{N^2(1+4\epsilon^2)}{4T\delta}\rightarrow 0 \text{ as } T\rightarrow \infty
\]
As a result, the spectral norm $\|\bar{G}-CMC'\|\leq \|\bar{G}-CMC'\|_F$ goes to 0 too. 
Based on lemma A.2 by \citet{oliveira2009concentration}, if $M$ has $K$ distinct eigenvalues, then the eigenvectors of $\bar{G}$ are close to the corresponding eigenvectors of $CMC'$. 
That is, let $u_i$ be the eigenvector corresponding to the $i$th largest eigenvalues of $\bar{G}$. 
Let $\theta_i$ be the counterpart for $CMC'$. 
If $\|\bar{G}-CMC'\|<\epsilon$, then $\|u_i u_i^T - \theta_i \theta_i^T\|<\delta \epsilon$. 
This implies that $1-(u_i^T\theta_i)^2<\delta\theta$. 
That is, $u_i$ is close to $\theta_i$ or $-\theta_i$. 
But $CMC'$ has only $K$ distinct rows. 
So the results show a spectral clustering on $\bar{G}$ will eventually lead to perfect recovery of the class labels.
\end{proof}

\section{Proof of Theorem \ref{theorem:consistency}}
\label{sec:ProofML}
We begin with the proofs of Lemmas \ref{lemma:KLdiffernece}--\ref{lemma:profile-mle}.

\begin{proof}[Proof of Lemma \ref{lemma:KLdiffernece}]
This is from Lemmas A1 and A2 by \citet{choi2012stochastic}. The main arguments are as follows. 
$h$, the expectation of the log-likelihood, is always maximized at the true parameters.
For any partition of $P$, any refinement of the partition increases $h$.
For any label assignment $z$, we can find a refinement that has at least $r(z)/2$ pairs of nodes that connect to at least $\min_k n_k(c)$ of nodes that differ at least $\delta$ from the truth.
\end{proof}

\begin{proof}[Proof of Lemma \ref{lemma:sigma}]
Because of symmetry, we only consider $p\in (0,\frac{1}{2}]$.
Let $C_0=p/2$. Then $C_0<p<1-C_0$. Let region $C=[C_0,1-C_0]$. 
By the Chernoff bound, $P(|x-p|>\epsilon)\leq 2\exp(-2N\epsilon^2)$. 
Therefore, $P(x\not\in C)\leq 2\exp(-Np^2/2)$.
Let $E_C(x)=\sum_{x\in C} xp(x)$. 
The subscript $C$ denotes any operation restricted on region $C$. 
Define the following functions and constants:
\begin{align*}
&\sigma(p)=p\log(p)+(1-p)\log(1-p); & M_0&=\max_{p\in C} |\sigma(p)| = -\sigma(0.5)\leq 0.7\\
&\sigma'(p)=\log(p)-\log(1-p); &M_1&=\max_{p\in C} |\sigma'(p)|=\log(1-C_0)-\log(C_0)\\
&\sigma''(p)=\frac{1}{p}+\frac{1}{1-p}; &M_2&=\max_{p\in C} |\sigma''(p)|=\frac{1}{C_0}+\frac{1}{1-C_0}\\
&\sigma'''(p)=-\frac{1}{p^2}+\frac{1}{(1-p)^2}; &M_3&=\max_{p\in C} |\sigma'''(p)|=\frac{1}{C_0^2}-\frac{1}{(1-C_0)^2}\\
&\sigma^{(4)}(p)=\frac{1}{2p^3}+\frac{1}{2(1-p)^3}; & M_4&=\max_{p\in C} |\sigma^{(4)}(p)|=\frac{1}{2C_0^3}+\frac{1}{2(1-C_0)^3}
\end{align*}
We can get the following bounds:
\begin{gather*}
|E_{\bar{C}}\sigma(x)|\leq M_0P(x\not\in C)\leq 2\exp\left(-\frac{Np^2}{2}\right)\\
E(x-p)^3 =\frac{p(1-p)(1-2p)}{N^2}\leq \frac{1}{4N^2}\\
E(x-p)^4=\frac{p(1-p)^3+p^3(1-p)}{N^3}+\frac{3(N-1)p^2(1-p)^2}{N^3}\leq \frac{1}{2N^3}+\frac{1}{4N^2}
\end{gather*}
By Taylor expansion on region $C$,
\[
\sigma(x)=\sigma(p)+\sigma'(p)(x-p)+\frac{\sigma''(p)}{2}(x-p)^2+\frac{\sigma'''(p)}{6}(x-p)^3+R(x)
\]
\[
|R(x)| \leq \max_{x\in C} |\sigma^{(4)}(x)(x-p)^4/24|.
\]
Thus
\begin{align*}
&N[E(\sigma(x))-\sigma(p)]-\frac{1}{2} \\
=\;&NE\left[\sigma(x)-\sigma(p)-\sigma'(p)(x-p)-\frac{\sigma''(p)}{2}(x-p)^2\right]\\
\leq \;&N E_C \left[\sigma(x)-\sigma(p)-\sigma'(p)(x-p)-\frac{\sigma''(p)}{2}(x-p)^2\right] + N(2M_0+2M_1)P(x\not\in C)\\
\leq \;&N E_C \left[\frac{\sigma'''(p)}{6}(x-p)^3\right] + \max_{x\in C} |\sigma^{(4)}(x)(x-p)^4/24|+ N(2M_0+2M_1)P(x\not\in C)\\
\leq \;&\frac{M_3}{24N} + \frac{M_4}{24}\left(\frac{1}{2N^2}+\frac{1}{4N}\right)+ 2N\left(1+M_1+\frac{M_3}{6}\right)\exp\left(-\frac{Np^2}{2}\right)\\
\rightarrow \;&0 \text{ as } N\rightarrow \infty,
\end{align*}
which completes the proof.
\end{proof}

\begin{proof}[Proof of Lemma \ref{lemma:profile-mle}]
For any $z$, as $\bar{P}$ averages $P$, $C_0\leq \bar{P}_{kl}(z)\leq 1-C_0$. We have 
\[
g(z)-h(z) = \sum_{k\leq l} n_{kl}(z)[E(x_{kl}(z))-\sigma(\bar{P}_{kl})]
\]
where $$x_{kl}(z)=\frac{o_{kl}(z)}{n_{kl}(z)}\sim \frac{1}{n_{kl}} \text{Bin}(n_{kl},\bar{P}_{kl}).$$
By Lemma \ref{lemma:sigma},
\[
n_{kl}(z)[E(x_{kl}(z))-\sigma(\bar{P}_{kl})] \rightarrow \frac{1}{2} + O\left(\frac{1}{n_{kl}(z)}\right).
\]
Therefore
\[
g(z)-h(z) = \frac{K(K+1)}{4} + O\left(\frac{K^2}{m(z)}\right).
\]
\end{proof}

\begin{proof}[Proof of Theorem \ref{theorem:consistency}]
We want to show that there exists $\delta_0$ such that 
\begin{equation}
Ef(c)-Ef(z)\geq \delta_0  \label{eqn:diff}
\end{equation}
for all $z\neq c$.
Then by Bernstein's inequality, we have 
\[
\frac{1}{T}\left|\sum_t [f^t(z)-Ef^t(z)]\right|\rightarrow 0  \text{ as } T\rightarrow \infty.
\]
Therefore
\[
\frac{1}{T}\sum_t f^t(c) - \frac{1}{T}\sum_t f^t(z) \rightarrow \frac{1}{T}\sum_t(Ef(c)-Ef(z))\geq \delta_0
\]
for all $z\neq c$.
Then we get the conclusion that $c$ is the unique maximizer of $\sum_t f^t(z)$.
To show \eqref{eqn:diff}, we know
\begin{align*}
Ef(c) - Ef(z) &=g(c)-g(z)\\
&=(h(c) - h(z)) + (g(c) - h(c)) - (g(z) - h(z))\\
&\geq \delta m(c) r(z) + (g(c) - h(c)) - (g(z) - h(z)).
\end{align*}
by Lemma \ref{lemma:KLdiffernece}.

Let $n_0$ denote the threshold that, for all $N \geq n_0$,  $|g(z)-h(z)-K(K+1)/4| \leq  \delta_1$ if $m(z)\geq n_0$.
Then for $m(c)\geq n_0$,
\[
g(c) - h(c)\geq \frac{K(K+1)}{4} - \delta_1 .
\]
For any $z$, the total number of nodes are $N\geq K m(c)$. The total number of nodes that do not satisfy $n_k(z)\geq n_0$ is at most $n_0(K-1)$. And for the rest of the nodes, we still have the bounded feature as in Lemma \ref{lemma:sigma}. Hence
\[
g(z)-h(z) \leq \delta_2 n_0 (K-1) + \frac{K(K+1)}{4} + \delta_1.
\]
Therefore
\[
Ef(c) - Ef(z) \geq \delta m(c) r(z) - 2\delta_1 -  \delta_2 n_0 (K-1).
\]
This is an increasing function of $m(c)$. So we can find large enough $m(c)$ that $\delta m(c) r(z) - 2\delta_1 -  \delta_2 n_0 (K-1) \geq \delta_0$ as required.
\end{proof}

\section{Minimum number of nodes for consistency with 2 classes}
\label{sec:MinNodes}
Theorem \ref{theorem:consistency} guarantees consistency of the MLE provided the conditions on $C_0$ and $\delta$ are satisfied, and the minimum number of nodes in any class is large enough. 
For the special case of $K=2$ classes, we can calculate the minimum number of nodes $N$ required to guarantee consistency. 
We need $N$ sufficiently large so that the expectation of the log-likelihood under the true class labels $c$ at each layer is larger than the expectation of the log-likelihood under any other class assignment $z$. 
With $2$ classes, for any given value of $N$, we can simply enumerate over the number of misclassified nodes to determine if this is indeed true. 
It suffices only to check two boundary cases for the other class assignment $z$:
\begin{itemize}
\item $2$ nodes are in one class, and the remaining $N-2$ nodes are in the other class.
\item Only a single node is misclassfied in $z$, i.e.~$c_i = z_i$ for all except a single value of $i \in \{1,\ldots,N\}$.
\end{itemize}
If the expectation of the log-likelihood under $c$ is indeed larger than the expectation of the log-likelihood under any other class assignment $z$ for both of these cases, then consistency as $T \rightarrow \infty$ is guaranteed for this value of $N$. 
If it is not, then one can simply iterate over values of $N$ until it is large enough to guarantee sufficiency. 

\section{Details of variational approximation}
\label{sec:Variational}
Let vector $\vec{z}_i=(z_{i1},\ldots,z_{iK})$ denote the class assignment vector for each node $i$.
\[
z_{ik}=\begin{cases} 
	1 & \text{if $i$ in class $k$}\\
	0 & \text{otherwise}. 
	\end{cases} 
\]
So $\vec{z}_i$ has all zeros except a single one indicating its class. 
This notation is easier for writing down the likelihood.
We denote the initial class assignment probability by $\vec{\pi}=(\pi_1,\ldots,\pi_K)$.
This is the multinomial parameter of $\vec{z}_i$.
The likelihood is
\[
l=\prod_{i,k} \pi_k^{z_{ik}} \prod_{i<j,k\leq l,t} [(p_{kl}^t)^{g_{ij}^t}(1-p_{kl}^t)^{1-g_{ij}^t}]^{z_{ik}z_{jl}}.
\]
It is difficult to maximize because $\vec{z}$ cannot be integrated out.
Instead we use variational approximation to decompose the likelihood into independent marginal distributions and apply an expectation-maximization technique to search for the maximum \citep{daudin2008mixture}.
Let $\vec{z}_i$ follow independent multinomial distributions, i.e.
\[
\vec{z}_i \stackrel{\text{ind}}{\sim} \text{Multi} (b_{i1},\ldots,b_{iK}), \quad E[z_{ik}]=b_{ik}.
\]
In the variational E-step, the approximate marginal distribution $q(z_i)$ is 
\begin{equation} \label{eqn:e-step}
\ln q(z_i) = \sum_k z_{ik}\left( \ln \pi_k  + \sum_{j\neq i,l,t}E[z_{jl}]\left(g_{ij}^t\ln P^t_{kl}+(1-g_{ij}^t)\ln (1-P^t_{kl})\right)\right) + \text{Const}.
\end{equation}
We update $b_{ik}$ according to \eqref{eqn:e-step}.
That is,
\begin{equation*}
b_{ik}\propto \pi_k \prod_{j\neq i}\prod_t \prod_l \left[(P_{kl}^t)^{g_{ij}^t}(1-P_{kl}^t)^{1-g_{ij}^t}\right]^{b_{jl}}
\end{equation*}
In the M-step, we maximize $\vec{\pi}$ and $P^t$ by
\begin{gather*}
\pi_k\propto\sum_i E[z_{ik}]=\sum_i b_{ik}\\
P^t_{kl}=\frac{\sum_{i\neq j}E[z_{ik}]E[z_{jl}]g_{ij}^t}{\sum_{i\neq j} E[z_{ik}]E[z_{jl}]}=\frac{\sum_{i\neq j}b_{ik}b_{jl}g_{ij}^t}{\sum_{i\neq j}b_{ik}b_{jl}}.
\end{gather*}
We iterate between the two steps until convergence.

\section*{Acknowledgments} 
We would like to thank Matteo Magnani for providing us with the AU-CS multi-layer network data. 
We also thank Brian Eriksson as well as the anonymous reviewers for their suggestions to improve the manuscript.

\bibliographystyle{abbrvnat}
\bibliography{reference}

\end{document}